\documentclass[english,runningheads,10pt]{llncs}

\usepackage{tabularx,booktabs,multirow,delarray,array}
\usepackage{graphicx,amssymb,amsmath,amssymb}
\usepackage{enumerate}
\usepackage{algorithmic}
\usepackage[ruled]{algorithm2e}
\usepackage{wrapfig}
\usepackage{latexsym}

\newtheorem{propo}{Proposition}
\newtheorem{cor}{Corollary}

\newcommand{\ignore}[1]{ }

\newcounter{rem}
\def\calR{\mathcal{R}}

\def\qed{\hbox{\rlap{$\sqcap$}$\sqcup$}}

\newenvironment{proof}{\par\noindent{\bf Proof:}}{\mbox{}\hfill$\qed$\\}

\begin{document}

\title{Dynamic algorithms for visibility polygons in simple polygons}

\author{
R. Inkulu\inst{1}\thanks{This research is supported in part by NBHM grant 248(17)2014-R\&D-II/1049}
\and 
K. Sowmya\inst{1}
\and 
Nitish P. Thakur\inst{1}
}

\institute{
Department of Computer Science \& Engineering\\
IIT Guwahati, India\\
\email{\{rinkulu,k.sowmya,tnitish\}@iitg.ac.in}
}

\maketitle

\pagenumbering{arabic}
\setcounter{page}{1}

\begin{abstract}
We devise the following dynamic algorithms for both maintaining as well as querying for the visibility and weak visibility polygons amid vertex insertions and/or deletions to the simple polygon.
\begin{itemize}
\vspace{0.05in}
\item A fully-dynamic algorithm for maintaining the visibility polygon of a fixed point located interior to the simple polygon amid vertex insertions and deletions to the simple polygon.
The time complexity to update the visibility polygon of a point $q$ due to the insertion (resp. deletion) of vertex $v$ to (resp. from) the current simple polygon is expressed in terms of the number of combinatorial changes needed to the visibility polygon of $q$ due to the insertion (resp. deletion) of $v$.
\vspace{0.03in}
\item An output-sensitive query algorithm to answer the visibility polygon query corresponding to any point $p$ in $\mathbb{R}^2$ amid vertex insertions and deletions to the simple polygon.
If $p$ is not exterior to the current simple polygon, then the visibility polygon of $p$ is computed.
Otherwise, our algorithm outputs the visibility polygon corresponding to the exterior visibility of $p$.
\vspace{0.03in}
\item An incremental algorithm to maintain the weak visibility polygon of a fixed-line segment located interior to the simple polygon amid vertex insertions to the simple polygon.
The time complexity to update the weak visibility polygon of a line segment $pq$ due to the insertion of vertex $v$ to the current simple polygon is expressed in terms of the sum of the number of combinatorial updates needed to the geodesic shortest path trees rooted at $p$ and $q$ due to the insertion of $v$.
\vspace{0.03in}
\item An output-sensitive algorithm to compute the weak visibility polygon corresponding to any query line segment located interior to the simple polygon amid both the vertex insertions and deletions to the simple polygon.
\vspace{0.1in}
\end{itemize}

Each of these algorithms requires preprocessing the initial simple polygon.
And, the algorithms that maintain the visibility polygon (resp. weak visibility polygon) compute the visibility polygon (resp. weak visibility polygon) with respect to the initial simple polygon during the preprocessing phase.

\end{abstract}


\section{Introduction}
\label{sect:intro}

Let $P$ be a simple polygon with $n$ vertices.
Two points $p, q \in P$ are said to be mutually {\it visible} to each other whenever the interior of line segment $pq$ does not intersect any edge of $P$.
For a point $q \in P$, the \emph{visibility polygon} $VP(q)$ of $q$ is the maximal set of points $x \in P$ such that $x$ is visible to $q$. 
The problem of computing the visibility polygon of a point in a simple polygon was first attempted in \cite{journals/cgip/Davis79}, who presented an $O(n^2)$ time algorithm.
Then, ElGindy and Avis \cite{journals/jal/ElGindyA81} and Lee \cite{journals/cvgip/Lee83} presented an $O(n)$ time algorithms for this problem.
Joe and Simpson \cite{journals/bitnummath/joe87} corrected a flaw in \cite{journals/jal/ElGindyA81,journals/cvgip/Lee83} and devised an $O(n)$ time algorithm that correctly handles winding in the simple polygon. 
For a polygon with holes, Suri et~al. \cite{conf/compgeom/SuriO86} devised an $O(n\lg{n})$ time algorithm.
An optimal $O(n+h\lg{h})$ time algorithm was given in Heffernan and Mitchell \cite{journals/siamcomp/HeffernanM95}.
Algorithms for visibility computation amid convex sets were devised in Ghosh \cite{journals/jal/Ghosh91}.
The preprocess-query paradigm based algorithms were studied in \cite{journals/siamcomp/GuibasMR97,journals/comgeo/BoseLM02,journals/dcg/AronovGTZ02,journals/algorithmica/AsanoAGHI86,conf/swat/Vegter90,conf/compgeom/ZareiG05,journals/comgeo/InkuluK09,journals/comgeo/ChenW15a,journals/comgeo/ChenW15}.
Algorithms for computing visibility graphs were given in \cite{journals/siamcomp/GhoshM91}.
For a line segment $pq \in P$, the {\it weak visibility polygon} $WVP(pq)$ is the maximal set of points $x \in P$ such that $x$ is visible from at least one point belonging to line segment $pq$.
Chazelle and Guibas \cite{journals/dcg/ChazelleG89}, and Lee and Lin \cite{journals/cvgip/LeeL86} gave an $O(n\lg{n})$ time algorithms for computing the weak visibility polygon of a line segment located interior to the given simple polygon.
Later, Guibas et~al. \cite{journals/algorithmica/GuibasHLST87} gave an $O(n)$ time algorithm for the same.
The query algorithms for computing weak visibility polygons were devised in \cite{journals/comgeo/BoseLM02,conf/cccg/BygiG11,journals/dcg/AronovGTZ02}.
Ghosh \cite{books/visalgo/skghosh2007} gives a detailed account of visibility related algorithms.
Given a simple polygon $P$ and a point $p$ (resp. line segment) interior to $P$, algorithms devised in \cite{conf/caldam/Inkulu17} maintain the visibility polygon of $p$ (resp. weak visibility polygon of $p$) as vertices are added to $P$. 
To our knowledge, \cite{conf/caldam/Inkulu17} gives the first dynamic (incremental) algorithms in the context of maintaining the visibility and weak visibility polygons in simple polygons.

\subsubsection*{Our contribution}

In the context of computing the visibility polygon (resp. weak visibility polygon), an algorithm is termed {\it fully-dynamic} when the algorithm maintains the visibility polygon (resp. weak visibility polygon) of a fixed point (resp. fixed-line segment) as the simple polygon is updated with vertex insertions and vertex deletions.
When a new vertex is added to the current simple polygon or when a vertex of the current simple polygon is deleted, the visibility polygon (resp. weak visibility polygon) of a fixed point (resp. fixed line segment) is updated.
A dynamic algorithm is said to be {\it incremental} if it updates the visibility polygon (resp. weak visibility polygon) of a fixed point (resp. fixed line segment) amid vertex insertions.
A dynamic algorithm is said to be {\it decremental} if it updates the visibility polygon (resp. weak visibility polygon) of a fixed point (resp. fixed line segment) amid vertex deletions.
We devise the following dynamic algorithms for maintaining as well as querying for the visibility polygons as well as weak visibility polygons.

\begin{itemize}
\item[*]
Our first algorithm (refer Section~\ref{sect:maintvpsimppolyint}) is fully-dynamic, and it maintains the visibility polygon of a fixed point $q$ located interior to the given simple polygon.
We preprocess the initial simple polygon $P$ having $n$ vertices to build data structures of size $O(n)$, and the visibility polygon of $q$ in $P$ is computed using the algorithm from \cite{journals/jal/ElGindyA81,journals/cvgip/Lee83,journals/bitnummath/joe87}.
When a vertex $v$ is added to (resp. deleted from) the current simple polygon $P'$, we update the visibility polygon of $q$ in $O((k+1)(\lg{n'})^2)$ time.
Here, $k$ is the number of updates required to the visibility polygon of $q$ due to the insertion (resp. deletion) of $v$ to (resp. from) $P'$.
(Both the insertion and deletion algorithms take $O((\lg{n'})^2)$ time to update even when $k$ is zero, hence the stated time complexity.)
We are not aware of any work other than \cite{conf/caldam/Inkulu17} to dynamically update the visibility polygon.
The algorithm in \cite{conf/caldam/Inkulu17} is an incremental algorithm to maintain the visibility polygon amid vertex insertions to simple polygon.
After preprocessing $P$ in $O(n)$ time, it takes $O((k+1)\lg{n'})$ time to update the visibility polygon due to the insertion of vertex $v$ to simple polygon, where $k$ is the number of combinatorial changes required to visibility polygon being updated due to the insertion of $v$ and $n'$ is the number of vertices of the current simple polygon. 

\vspace{0.1in}

\item[*]
Our second algorithm (refer Section~\ref{sect:maintvpsimppolyext}) answers the visibility polygon of any query point located in $\mathbb{R}^2$ in $O(k(\lg{n'})^2)$ time, amid vertex insertions and deletions to the simple polygon. 
Here, $k$ is the number of vertices of $VP(q)$, and $n'$ is the number of vertices of the current simple polygon.
This algorithm computes $O(n)$ sized data structures by preprocessing the initial simple polygon $P$ defined with $n$ vertices.

\vspace{0.1in}

\item[*]
Our third algorithm (refer Section~\ref{sect:wvpsimppolyincr}) maintains the weak visibility polygon of a fixed line segment $pq$ located interior to the given simple polygon amid vertex insertions to the simple polygon.
It preprocesses the initial simple polygon $P$ defined with $n$ vertices in $O(n)$ time to build data structures of size $O(n)$.
The preprocessing time includes the time to compute the visibility polygon of $q$ in $P$ using the algorithm from \cite{journals/algorithmica/GuibasHLST87}.
When a vertex $v$ is added to the current simple polygon $P'$, we update the weak visibility polygon of $pq$ in $O((k+1)\lg{n'})$ time.
Here, $k$ is the sum of changes in the number of updates required to shortest path trees rooted at $p$ and $q$ due to the insertion of $v$, and $n'$ is the number of vertices of $P'$.
(Similar to the update complexity of the first algorithm, incremental algorithm devised here takes $O(\lg{n'})$ time to update the weak visibility polygon of $pq$ even when $k$ is zero, hence the stated time complexity.)

\vspace{0.1in}

\item[*]
Our fourth and final algorithm (refer Section~\ref{sect:wvpsimppolyint}) answers the weak visibility polygon query of any line segment located interior to the current simple polygon $P'$ in $O(k(\lg{n'})^2)$ time amid vertex insertions and vertex deletions to simple polygon after preprocessing the initial simple polygon $P$ in $O(n)$ time.
Here, $k$ is the output complexity, $n$ is the number of vertices of $P$, and $n'$ is the number of vertices of $P'$. 
\end{itemize}

To our knowledge, the fully-dynamic algorithm to maintain the visibility polygon of a fixed point located interior to the simple polygon is the first fully-dynamic algorithm devised in this context. 
And, same is the case with the query algorithms devised to maintain the visibility and weak visibility polygons amid vertex insertions and deletions.
The preliminary version of weak visibility polygon maintenance algorithm was published in a conference \cite{conf/caldam/Inkulu17}, which is detailed herewith.
Moreover, we devise an algorithm to compute the visibility polygon of a point located in the given simple polygon using ray-shooting and ray-rotating queries, which could be of independent interest.

We assume that after adding or deleting any vertex of the current simple polygon, the polygon remains simple. 
Moreover, it is assumed that every new vertex is added between two successive vertices of the current simple polygon.
In maintaining the visibility polygon (resp. weak visibility polygon), the point $q$ (resp. line segment $l$) whose visibility polygon (resp. weak visibility polygon) is updated remains interior to the updated simple polygon if it was interior to the simple polygon before the vertex insertion/deletion; and, $q$ (resp. $l$) remains exterior if it was exterior to $P$.
The initial simple polygon is denoted with $P$.
We use $P'$ to denote the simple polygon just before inserting/deleting a vertex and $P''$ is the simple polygon after inserting/deleting a vertex.
Further, we assume that $P, P'$, and $P''$ are respectively defined with $n, n'$, and $n''$ vertices.
Whenever we delete a vertex $v$ of $P'$, which is adjacent to vertices $v_i$ and $v_{i+1}$ in $P'$, it is assumed that an edge is introduced between $v_i$ and $v_{i+1}$ after deleting $v$.
Similarly, whenever we insert a vertex $v$ between adjacent vertices $v_i$ and $v_{i+1}$ of $P'$, it is assumed that two edges are introduced: one between $v_i$ and $v$, and the other between $v_{i+1}$ and $v$. 
The boundary of a simple polygon $P$ is denoted with $bd(P)$.
Unless specified otherwise, the boundary of the simple polygon is assumed to be traversed in the counterclockwise direction.

Let $u_iu_{i+1}$ be an edge on the boundary of $VP(q)$ such that (i) no point of $u_iu_{i+1}$, except the points $u_i$ and $u_{i+1}$, belong to the boundary of $P$, and (ii) one of $u_i$ or $u_{i+1}$ is a vertex of $P$.
Then such an edge $u_iu_{i+1}$ is called a {\it constructed edge}.
For every constructed edge $u_iu_{i+1}$, among $u_i$ and $u_{i+1}$ the farther from $q$ is termed a {\it constructed vertex} of $VP(q)$.
The constructed edges of $VP(q)$ partition $P$ into a set $\calR = \{VP(q), R_1, R_2, \ldots, R_s\}$ of simple polygonal regions such that no point $p$ interior to region $R$ is visible from $q$ for any $R \in \calR$ and $R \ne VP(q)$.
These regions are termed {\it occluded regions} from $q$.
For each $R \in \calR$ and $R \ne VP(q)$, there exists a constructed vertex associated to $R$.
Since each vertex of $P$ may cause at most one constructed edge, there can be $O(n)$ constructed edges.
In the context of weak visibility polygons, both the constructed vertices and constructed edges are defined analogously. 
(Refer to \cite{books/visalgo/skghosh2007}).

Given a simple polygon $P$, the {\it ray-shooting query} of a ray $\overrightarrow{r} \in \mathbb{R}^2$ determines the first point of intersection of $\overrightarrow{r}$ with the $bd(P)$. 
Given two points $p'$ and $p''$ in the interior/exterior of a simple polygon $P$, the {\it shortest-distance query} between $p'$ and $p''$ outputs the geodesic Euclidean distance between $p'$ and $p''$.

An algorithm to compute the visibility polygon of a point interior to a given simple polygon using the ray-shooting and ray-rotating queries is described in Section~\ref{sect:rayshootrayrotate}.
Section~\ref{sect:maintvpsimppolyint} devises a fully-dynamic algorithm to maintain the visibility polygon of a point interior to the simple polygon. 
Further, in Section~\ref{sect:maintvpsimppolyint}, we devise an output-sensitive algorithm to answer the visibility polygon queries when the query point is interior to the simple polygon.
Section~\ref{sect:maintvpsimppolyext} devises an output-sensitive algorithm to answer visibility polygon queries when the query point is exterior to the simple polygon. 
Section~\ref{sect:wvpsimppolyincr} devises an incremental algorithm to maintain the weak visibility polygon of a fixed edge of the simple polygon.
An algorithm to query for the weak visibility polygon of a line segment when that line segment is interior to the simple polygon is presented in Section~\ref{sect:wvpsimppolyint}.
The conclusions are in Section~\ref{sect:conclu}.

\vspace{-0.12in}

\section{Ray-shooting and ray-rotating queries in computing the visibility polygon}
\label{sect:rayshootrayrotate}

In this section, we describe ray-shooting and ray-rotating queries from the literature. 
Using these query algorithms, we devise an algorithm to compute the visibility polygon of a point located in the given simple polygon.
We need this visibility polygon computation algorithm in devising dynamic algorithms in the next section.
Further, this algorithm could be useful on its own.

First, we state a Theorem from \cite{journals/jal/GoodrichT97} that facilitates in answering ray-shooting queries amid vertex insertions and deletions.

\begin{propo}[\cite{journals/jal/GoodrichT97}~Theorem~6.3]
\label{propo:rayshoot}
Let $\mathcal{T}$ be a planar connected subdivision with $n$ vertices.
With $O(n)$-time preprocessing, a fully dynamic data structure of size $O(n)$-space is computed for $\mathcal{T}$ that supports point-location, ray-shooting, and shortest-distance queries in $O((\lg{n})^2)$ time, and operations InsertVertex, RemoveVertex, InsertEdge, RemoveEdge, AttachVertex, and DetachVertex in $O((\lg{n})^2)$ time, all bounds being worst-case.
\end{propo}

Further, for any two query points $q', q''$ belonging to simple polygon $P$, we note that the fully dynamic data structures in \cite{journals/jal/GoodrichT97} support outputting the first line segment in the geodesic shortest-path from $q'$ to $q''$ in $O((\lg{n})^2)$ time.

Given a ray $\overrightarrow{r}$ whose origin $q$ belongs to simple polygon $P'$, the {\it ray-rotating query} (defined in \cite{journals/comgeo/ChenW15a}) with clockwise (resp. counterclockwise) orientation seeks the first vertex of $P'$ visible to $q$ that will be hit by $\overrightarrow{r}$ when we rotate $\overrightarrow{r}$ by a minimum non-negative angle in clockwise (resp. counterclockwise) direction. 
The first parameter to the ray-rotating-query algorithm is the ray, and the second one determines whether to rotate the input ray by a non-negative angle in clockwise or in the counterclockwise direction.

\begin{propo}[\cite{journals/comgeo/ChenW15a}~Lemma~1]
\label{propo:rayrot}
By preprocessing a simple polygon having $n$ vertices, a data structure can be built in $O(n)$ time and $O(n)$ space such that each ray-rotating query can be answered in $O(\lg{n})$ time.
\end{propo}

To support ray-rotating queries, \cite{journals/comgeo/ChenW15a} in turn uses ray-shooting data structure and two-point shortest distance query data structures from \cite{journals/jcss/GuibasH89}. 
But for supporting dynamic insertion (resp. deletion) of vertices to (resp. from) the given simple polygon, we use data structures from \cite{journals/jal/GoodrichT97} in place of the ones from \cite{journals/jcss/GuibasH89}.
This leads to preprocessing simple polygon $P$ defined with $n$ vertices in $O(n)$ time to compute data structures of $O(n)$ space so that ray-shooting, ray-rotating, and two-point distance queries are answered in $O((\lg{n})^2)$ worst-case time.

For any two arbitrary rays $r_1$ and $r_2$ with their origin at $q$, the $cone(r_1, r_2)$ comprises of a set $S$ of points in $\mathbb{R}^2$ such that $x \in S$ whenever a ray $r$ with origin at $q$ is rotated with center at $q$ from the direction of $r_1$ to the direction of $r_2$ in counterclockwise direction the ray $\overrightarrow{qx}$ occurs.
The $opencone(r_1,r_2)$ is the $cone(r_1,r_2) \setminus \{r_1,r_2\}$.

Let $P$ be a simple polygon and let $q$ be a point located in $\mathbb{R}^2$.
Also, let $r_1$ and $r_2$ be two rays with origin at $q$.
The $visvert$-$inopencone$ algorithm listed underneath outputs all the vertices of $P$ that are visible from $q$ in the $opencone(r_1, r_2)$.
This is accomplished by issuing a series of ray-rotating queries in the $opencone(r_1, r_2)$, essentially sweeping the $opencone(r_1, r_2)$ region to find all the vertices of $P$ that are visible from $q$.

\begin{figure}[h]
\centering
\includegraphics[totalheight=1.3in]{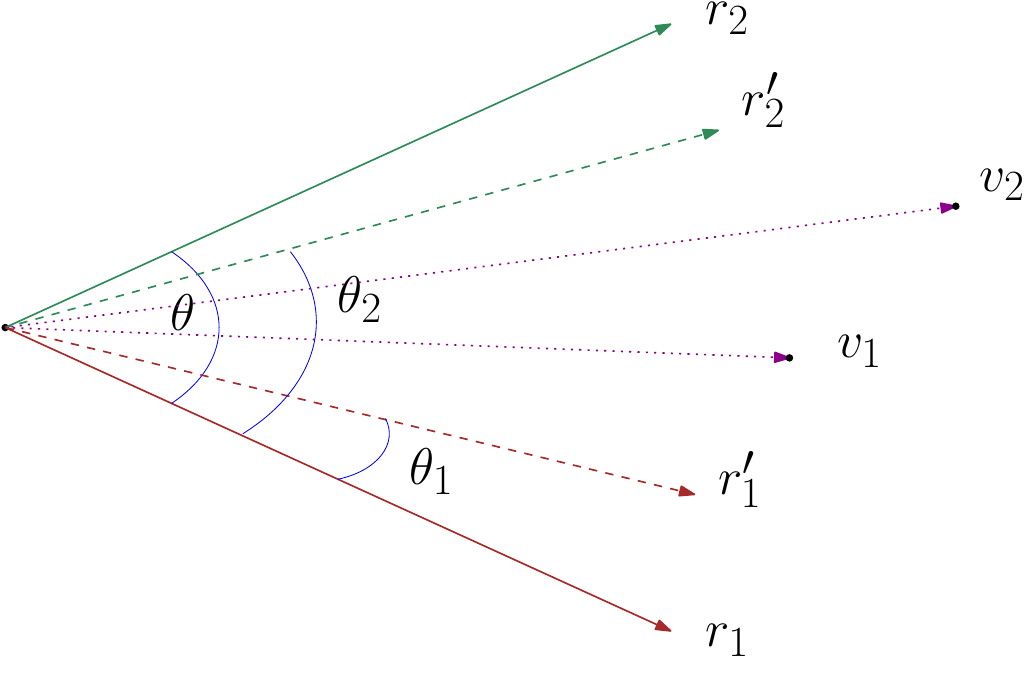}
\vspace{-0.15in}
\caption{\footnotesize Illustrating three cones, each correspond to a recursive call in $visvert$-$inopencone$ algorithm}
\label{fig:visvert}
\end{figure}

We compute the $VP(q)$ by invoking Algorithm~\ref{algo:visvert-inopencone} (listed underneath) with ray $r_1$ as both the first parameter as well as the second parameter.
This invocation yields all the vertices of $P$ that are visible from $q$ except the ones that may lie along the ray $r_1$.
And, one invocation of ray-shooting-query with ray $r_1$ outputs any vertex along the ray $r_1$ that is visible from $q$.

\hfil\break

\begin{algorithm}[H]
\label{algo:visvert-inopencone}

\SetKwInOut{Input}{Input}
\SetKwInOut{Output}{Output}

\caption{$visvert$-$inopencone(r_1, r_2)$}

\Input{A simple polygon $P$ and two rays $r_1$ and $r_2$ with their origin at a point $q \in P$}

\Output{vertices in $P \cap opencone(r_1, r_2)$ that are visible from $q$}

\begin{algorithmic}[1]

\STATE
$\theta := \cos^{-1}(\frac{r_1.r_2}{|r_1| |r_2|})$;
If $\theta$ equals to $0$, then $\theta := \pi$.
\STATE
Let $r_1'$ and $r_2'$ be the rays with origin at $q$ that respectively make $\theta_1$ and $\theta_2$ counterclockwise angles with ray $r_1$ such that $\theta_1 < \theta_2 < \theta$.
(Refer to Fig.~\ref{fig:visvert}.)

\STATE 
$v_1 :=$ ray-rotating-query$(r_1', counterclockwise)$
\STATE 
$v_2 :=$ ray-rotating-query$(r_2', clockwise)$
\STATE 
Output $v_1, v_2$

\STATE
$visvert$-$inopencone(r_1, r_1')$ 
\STATE 
$visvert$-$inopencone(r_2', r_2)$

\STATE If $v_1 \ne v_2$ then $visvert$-$inopencone(\overrightarrow{qv_1}, \overrightarrow{qv_2})$

\end{algorithmic}
\end{algorithm}

\hfil\break

We note that the ray-rotating-query from \cite{journals/comgeo/ChenW15a} works only if the input ray does not pass through a vertex of the simple polygon visible from that ray's origin. 
To take this into account, in Algorithm~\ref{algo:visvert-inopencone}, we perturb $r_1$ (resp. $r_2$) to obtain another ray $r_1'$ (resp. $r_2'$) such that $r_1' \in opencone(r_1, r_2)$ (resp. $r_2' \in opencone(r_1, r_2)$) and $q$ is the origin of $r_1'$ (resp. $r_2'$). 
After obtaining visible vertices $v_1$ and $v_2$ in steps (3) and (4) via ray-rotating queries, in steps (6), (7), and (8) of the Algorithm, we recursively compute vertices visible in three cones.
In addition, for every vertex $v$ visible from $q$, our algorithm shoots a ray $\overrightarrow{qv}$ to determine the possible constructed edge on which $v$ resides.

\begin{lemma}
\label{lem:visvertopenconecorr}
The $visvert$-$inopencone(r_1, r_2)$ algorithm outputs every vertex of $P$ in $opencone(r_1, r_2)$.
\end{lemma}
\begin{proof}
Consider any non-leaf node $b'$ of the recursion tree.
The open cone corresponding to $b'$ is divided into three open cones.
Assuming that the set $S$ of vertices of $P$ within these three open cones are computed correctly, the vertices in $S$ together with the vertices computed at node $b'$ of the recursion tree along with the rays that separate these three cones ensure the correctness of the algorithm.
\end{proof}

If there are $k$ vertices of $P$ that are visible from $q$, our algorithm takes $O(k(\lg{n})^2)$ time: it involves $O(k)$ ray-rotating queries, $O(k)$ ray-shooting queries to compute the respective constructed edges, and $O(k\lg{k})$ time to sort the vertices of $VP(q)$ according to their angular order. 

\begin{theorem}
\label{thm:visvertopencone}
Our algorithm preprocesses the given simple polygon $P$ defined with $n$ vertices in $O(n)$ time and computes $O(n)$ spaced data structures to facilitate in answering the visibility polygon $VP(q)$ of any given query point $q \in P$ in $O(k(\lg{n})^2)$ time.
Here, $k$ is the number of vertices of $VP(q)$.
\end{theorem}
\begin{proof}
The correctness is immediate from Lemma~\ref{lem:visvertopenconecorr}.
The preprocessing and query complexities are argued above.
\end{proof}

\section{Maintaining the visibility polygon of a fixed point $q \in P'$}
\label{sect:maintvpsimppolyint}

In this section, we consider the problem of updating the $VP(q)$ when a new vertex is inserted to the current simple polygon $P'$ or an existing vertex is deleted from $P'$, resulting in a new simple polygon $P''$.
The following simple data structures are used in our algorithms:

\begin{itemize}
\item[-] A circular doubly linked list $L_{P'}$, whose each node stores a unique vertex of $P'$. 
The order in which the vertices of $P'$ occur while traversing $bd(P')$ in the counterclockwise direction (starting from an arbitrary vertex of $P'$) is the order in which vertices occur while traversing $L_{P'}$ in counterclockwise direction.

\item[-] A circular doubly linked list $L_{vp}$, whose each node stores a unique vertex of the visibility polygon of a point $q$.
The order in which the vertices of $VP(q)$ occur while traversing $bd(VP(q))$ in the counterclockwise direction is the order in which vertices occur while traversing $L_{vp}$ in the counterclockwise direction.

\item[-]
A balanced binary search tree (such as the one given in \cite{books/dsnetworkalgo/tarjan1983}) $B_{P'}$ (resp. $B_{vp}$) with nodes of $L_{P'}$ (resp. $L_{vp}$) as leaves.
The left-to-right order of leaves in $B_{P'}$ (resp. $B_{vp}$) is same as the order in which the nodes of $L_{P'}$ (resp. $L_{vp}$) occur while traversing $L_{P'}$ (resp. $L_{vp}$) in counterclockwise order, starting from an arbitrary node of $L_{P'}$ (resp. $L_{vp}$).
Further, we maintain pointers between the corresponding leaf nodes of $B_{P'}$ and $B_{vp}$ representing the same vertex.

\item[-] For each edge $e$, an array is associated with it to save the constructed vertices that are incident on $e$.

\item[-] And, the data structures needed for the dynamic ray shooting and two-point shortest-distance queries from \cite{journals/jal/GoodrichT97}.
\end{itemize}

After every insertion as well as deletion of any vertex, we update the data structures relevant to ray-shooting queries.
As mentioned, ray-rotating queries mainly rely on ray-shooting query data structures.
Assuming that the current simple polygon $P'$ is defined with $n'$ vertices, due to Proposition~\ref{propo:rayshoot}, these updates take $O((\lg{n'})^2)$ time.

\section*{Inserting a vertex}
\label{subsect:insertvertintsimppoly}

Let $v$ be the vertex being inserted to the current simple polygon $P'$.
And, let the edge $v_iv_{i+1}$ of $P'$ be replaced with edges $v_iv$ and $vv_{i+1}$, resulting in a new simple polygon $P''$.
Our algorithm handles the following two cases independently: (i) vertex $v$ is visible from $q$ in $P''$, (ii) $v$ is not visible from $q$ in $P''$. 
With the ray-shooting query with ray $\overrightarrow{qv}$, we check whether this ray strikes $bd(P'')$ at $v$.
If it is, then the vertex $v$ is visible from $q$ in $P''$.
Otherwise, it is not.

\begin{figure}
\centering
\includegraphics[totalheight=1.2in]{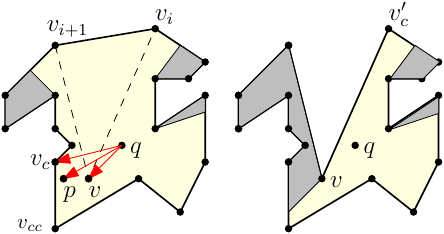}
\vspace{-0.15in}
\caption{\footnotesize Illustrating the case in which the inserted vertex $v$ is visible from $q$ in the new simple polygon $P''$; $VP(q)$ is shaded in light yellow; points interior to grey shaded polygonal regions are not visible from $q$.}
\label{fig:simppolyintinserti}
\end{figure}

In case~(i), using balanced binary search tree $B_{vp}$, we find two vertices $v_c, v_{cc}$ such that the ray $qv$ lies in the $cone(\overrightarrow{qv_c}, \overrightarrow{qv_{cc}})$ and $v_c, v_{cc}$ successively occur in that order while traversing $bd(VP(q))$ in counterclockwise direction. 
(Refer to Fig.~\ref{fig:simppolyintinserti}.)
We determine whether the triangle $qvv_{cc}$ or the triangle $qvv_c$ or both intersects with the triangle $v_ivv_{i+1}$.
Suppose the triangle $qvv_{c}$ intersects with the triangle $v_ivv_{i+1}$.
(The other two cases are handled analogously.)
We choose a point $p$ in the $opencone(\overrightarrow{qv}, \overrightarrow{qv_{c}})$ and invoke ray-rotate query with ray $qp$ in the clockwise direction in $P''$ to find the vertex $v_{c}'$ of $P''$. 

\begin{lemma}
\label{lem:contseqblock}
Let $VP(q)$ be the visibility polygon of a point $q \in P'$.
When a vertex $v$ is inserted to the boundary of a simple polygon $P'$, either no vertex in $VP(q)$ is hidden due to the insertion of $v$ or the vertices in $VP(q)$ that are hidden due to the insertion of $v$ are consecutive along the boundary of $VP(q)$. 
\end{lemma}
\begin{proof}
Let $v$ be the vertex inserted into $bd(P')$ between vertices $v_i$ and $v_{i+1}$.
A vertex $u$ of $VP(q)$ is not visible due to the insertion of $v$ whenever the line segment $qu$ intersects with the triangle $v_ivv_{i+1}$.
If the triangle $vv_iv_{i+1}$ does not intersect with $VP(q)$, then no vertex in $VP(q)$ gets hidden due to the insertion of $v$ to $P'$.
Without loss of generality, suppose the ray $qv_{i+1}$ makes a larger angle with ray $qv$ than does the ray $qv_i$.
Consider the cone $C_q$ with apex at $q$, bounded by rays $qv$ and $qv_{i+1}$, and the ray $qv_i$ intersecting the interior of $C_q$. 
We note that the set of rays from $q$ to vertices of $VP(q)$ that intersect cone $C_q$ are contiguous among all the rays that originate from $q$.
\end{proof}

Since every vertex that occur while traversing $bd(VP(q))$ from $v_{c}$ to $v_{c}'$ in clockwise direction is not visible from $q$, we delete all of these vertices from $B_{vp}$.
Due to Lemma~\ref{lem:contseqblock}, no other vertex needs to be deleted from $VP(q)$.
Further, we include $v$ between nodes $v_c$ and $v_{cc}$ in $L_{vp}$, and by using the appropriate keys we insert $v$ into $B_{P'}$ so that $B_{vp}$ is a balanced binary search tree over the nodes of $L_{vp}$. 
Similarly, we insert $v$ into $L_{P'}$ and $B_{P'}$.
Using ray-shooting queries, we compute the constructed edges that could be incident to $v$.
And, we include the new constructed vertex into both $L_{vp}$ and $B_{vp}$.

For the Case~(ii), the vertex $v$ is not visible from $q$.
(Refer to Fig.~\ref{fig:simppolyintinsertii}.)
Using $B_{vp}$, we determine the set $R$ of points of intersection of the line $l_{vv_i}$ that supports the line segment $vv_i$ with $VP(q)$.
Among all the points in $R$, let $p$ be the point that is at a larger distance from $v_i$.
The distance between $v_i$ to $p$ in comparison to the distance between $v_i$ and $v$ determines whether $v$ is in an occluded region.
Suppose for two constructed vertices $v_1', v_2'$ of $VP(q)$, the constructed edges $v_1'v_1''$ and $v_2'v_2''$ along $\overrightarrow{qv_1'}$ and $\overrightarrow{qv_2'}$ intersect triangle $v_ivv_{i+1}$.
When such pairs do not exist, it means that no section of the boundary of the triangle $v_ivv_{i+1}$ is visible from $q$, and $v$ is in a region whose interior is occluded from $q$.
(Note that the algorithm described herewith to handle Case~(ii) can be simplified to work when no constructed edge intersects with $v_iv$ or $v_{i+1}v$ as well.)

\begin{figure}[h]
\centering
\includegraphics[totalheight=1.3in]{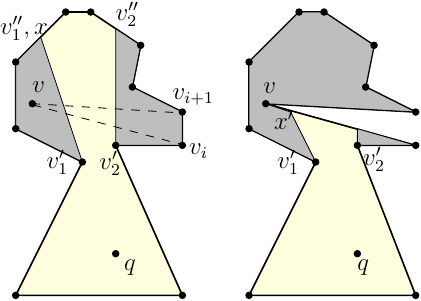}
\vspace{-0.15in}
\caption{\footnotesize Illustrating Case (ii) in which the inserted vertex $v$ is not visible from $q$ in the new simple polygon $P''$.}
\label{fig:simppolyintinsertii}
\end{figure}

Otherwise, let $v_1'x$ be the constructed edge that intersects $vv_{i}$.
(Analogous cases are handled similarly.)
We replace the constructed edge $v_1'x$ with $v_1'x'$, where $x'$ is the point of intersection of $v_1'x$ with $v_iv$ by including $x'$ into data structures $L_{vp}$ and $B_{vp}$.
Since no vertex in the clockwise traversal of $bd(VP(q))$ from $v_1''$ to $v_2''$ is visible from $q$ due to triangle $v_ivv_{i+1}$, we remove each of these vertices from both $L_{vp}$ and $B_{vp}$.
Further, the constructed edges of $v_i$ and $v_{i+1}$ are updated if necessary.

In both the cases, for every constructed vertex $p$ that is incident to edge $v_iv_{i+1}$ of $P'$, let $v_j$ be the vertex of $P'$ that is incident to line segment $qp$. 
We ray-shoot with ray $\overrightarrow{v_jp}$ in $P''$ to find the new constructed vertex that is incident to either the edge $v_iv$ or the edge $vv_{i+1}$ of $P''$.
(Refer to Fig.~\ref{fig:incrconstrvert}.)

\begin{figure}[h]
\centering
\includegraphics[totalheight=1.2in]{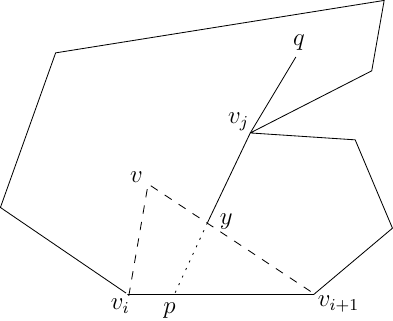}
\vspace{-0.15in}
\caption{\footnotesize Illustrating the modification of constructed vertices that lie on $v_iv_{i+1}$ due to the insertion of vertex $v$.}
\label{fig:incrconstrvert}
\end{figure}

In updating the $VP(q)$, there are $O(1)$ ray-shooting and ray-rotating queries involved.
It takes $O(\lg{n'})$ time per deletion of a vertex from $B_{P'}$ or $B_{vp}$ and $O(1)$ time per deletion of a vertex from $L_{P'}$ or $L_{vp}$.
Here, $n'$ is the number of vertices of the current simple polygon $P'$.
Hence, the time complexity of our algorithm in updating the visibility polygon due to the insertion of $v$ is $O((k+1)(\lg{n'})^2)$, where $k$ is the number of vertices that required to be added and/or removed from the visibility polygon of $q$ due to the insertion of $v$ to $P'$.

\begin{lemma}
\label{lem:2b}
Let $P'$ be a simple polygon.
Let $v$ be the vertex inserted to $bd(P')$, resulting in a simple polygon $P''$.
Let $q$ be a point belonging to both $P'$ and $P''$.
Also, let $VP_{P'}(q)$ be the visibility polygon of $q$ in $P'$, and let $VP_{P''}(q)$ be the updated visibility polygon of $q$ computed using our algorithm.
A point $p \in VP_{P''}(q)$ if and only if $p$ is visible from $q$ in $P''$.
\end{lemma}
\begin{proof}
When a vertex is added to $P'$, the only vertices that are added to $VP_{P'}(q)$ are the vertices of the constructed edges computed using ray-shooting queries. 
Since the points returned by the ray-shooting queries are always visible from $q$, the newly added vertices to the visibility polygon of $q$ in $P'$ are visible from $q$.
From Lemma~\ref{lem:contseqblock}, visibility is blocked only for one particular consecutive set of vertices in $VP_{P'}(q)$; every such vertex is removed using ray-rotating query in $P''$.
In other words, a point $p$ of $P''$ is a vertex of $VP_{P''}(q)$ whenever $p$ is visible from $q$ in $P''$.
\end{proof}

\subsection*{Deleting a vertex}
\label{subsect:deletevertintsimppoly}

Let $v$ be the vertex to be deleted from the current simple polygon $P'$, and $v_iv$ and $vv_{i+1}$ be the edges of $P'$.
Also, let $v_iv$ and $vv_{i+1}$ occur in that order while traversing $bd(P')$ in counterclockwise direction starting from $v_i$.
Also, let $P''$ be the resultant simple polygon due to the deletion of vertex $v$ from $P'$ (and adding the edge $v_iv_{i+1}$).
Our algorithm handles the following two cases independently: (i) vertex $v$ is visible from $q$ in $P'$, (ii) $v$ is not visible from $q$ in $P'$. 

\begin{figure}
\centering
\includegraphics[totalheight=1.1in]{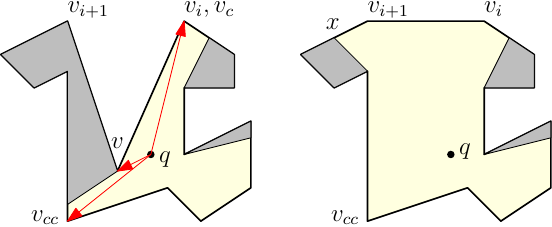}
\vspace{-0.15in}
    \caption{\footnotesize Illustrating Case (i) in which the deleted vertex $v$ is visible from $q$ in simple polygon $P'$.}
\label{fig:simppolyintdeli}
\end{figure}

\begin{lemma}
\label{lem:contseqvisdel}
Let $P'$ be  a simple polygon.
Let $P''$ be the simple polygon obtained by deleting a vertex $v$ from $P'$.
Also, let $VP(q)$ be the visibility polygon of a point $q$ interior to $P''$.
The set of vertices that become visible due to the deletion of vertex $v$ from $P'$ are consecutive along the boundary of $VP(q)$.
\end{lemma}
\begin{proof}
Immediate from Lemma~\ref{lem:contseqblock}.
\end{proof}

A ray-shooting query with ray $qv$ determines whether the vertex $v$ is visible from $q$.
If $v$ is visible from $q$ in $P'$, using $B_{vp}$, we find the predecessor vertex $v_c$ and the successor vertex $v_{cc}$ of $v$ that occur while traversing $bd(VP(q))$ in counterclockwise direction.
Then we invoke $visvert$-$inopencone$ algorithm twice with simple polygon $P''$: once with rays $qv$ and $qv_{cc}$ as respective first and second parameters; next with rays $qv_c$ and $qv$ as respective first and second parameters. 
The vertices output by these invocations are precisely the ones in $P'$ hidden from $q$ due to triangle $v_ivv_{i+1}$.
If $k$ vertices become visible from $q$ after the deletion of $v$ from $P'$, then this sub-procedure requires $O(k)$ ray-rotations. 
We compute the constructed edges corresponding to each vertex that gets visible from $q$ after the removal of $v$.
For every constructed vertex $p$ that is incident to either of the edges $vv_i$ or $vv_{i+1}$ of $P'$, let $v_j$ be the vertex of $P$ that is incident to line segment $qp$. 
We ray-shoot with ray $v_jp$ in $P''$ to find the new constructed edge on which the vertex $v_j$ is incident.
Also, we do the corresponding updates to data structures $B_{vp}$ and $B_{P'}$.

\begin{figure}[h]
\centering
\includegraphics[totalheight=1.3in]{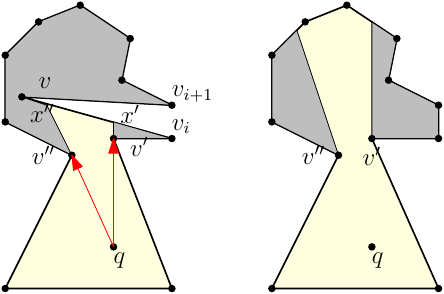}
\vspace{-0.15in}
\caption{\footnotesize Illustrating Case (ii) in which the deleted vertex $v$ is not visible from $q$ in simple polygon $P'$.}
\label{fig:simppolyintdelii}
\end{figure}

For the case~(ii), if neither $vv_i$ nor $vv_{i+1}$ has constructed vertices of $VP(q)$ stored with them, then we do nothing.
Otherwise, let $x', x''$ be the constructed vertices stored with the edge $vv_i$ such that $x''$ occurs after $x'$ in the counterclockwise ordering of the vertices of $VP(q)$. 
Refer to Fig.~\ref{fig:simppolyintdelii}.
(Handling the possible constructed vertices that could be incident to $vv_{i+1}$ are analogous.)
To find the set $S$ of vertices of $P'$ that are visible from $q$, we invoke $visvert$-$inopencone$ algorithm with rays $qx'$ and $qx''$ respectively as the first and second parameters.
For every vertex $v_l \in S$, using ray-shooting query, we determine the constructed edge corresponding to $v_l$.
Let $v'$ (resp. $v''$) be the vertex that lie on constructed edge $qx'$ (resp. $qx''$).
With two ray-shooting queries in $P'$, one with ray $v'x'$ and the other with ray $v''x''$, we find the new constructed edges on which $v'$ and $v''$ respectively lie. 

Let $k$ be the number of updates required to the visibility polygon of $q$ due to the deletion of vertex $v$ from $P'$.
In both the cases, the time complexity is dominated by $k$ ray-rotating queries, which together take $O((k+1)(\lg{n'})^2)$ time. 
(Since the deletion algorithm takes $O((\lg{n'})^2)$ time to update even when $k$ is zero, we write the update time complexity as $O((k+1)(\lg{n'})^2)$ per deletion.)

\begin{lemma}
\label{lem:2c}
Let $P'$ be a simple polygon.
Let $v$ be the vertex deleted from $bd(P')$, resulting in a simple polygon $P''$.
Let $q$ be a point belonging to both $P'$ and $P''$.
Also, let $VP_{P'}(q)$ be the visibility polygon of $q$ in $P'$, and let $VP_{P''}(q)$ be the updated visibility polygon of $q$ computed using our algorithm.
A point $p \in VP_{P''}(q)$ if and only if $p$ is visible from $q$ in $P''$.
\end{lemma}
\begin{proof}
Let $v_c$ and $v_{cc}$ be the vertices as defined in the description of the algorithm. 
The vertices between $v_c$ and $v_{cc}$ are added to $VP_{P'}(q)$.
These vertices were occluded prior to the deletion of $v$.
Along with these vertices, constructed vertices corresponding to these vertices are also included into $VP_{P'}(q)$.
Due to the Lemma~\ref{lem:contseqvisdel}, these are the only set of vertices of $VP_P(q)$ whose visibility from $q$ gets affected.
Let $V$ be the set comprising of constructed vertices that lie on $v_iv$, $v_{i+1}v$, together with endpoint of possible constructed edge that is incident to $v$.
Each vertex $v' \in V$ of $VP_{P'}(q)$ is replaced with a new vertex in $VP_{P''}(q)$ with a ray-shooting query; hence, the new constructed vertices are computed correctly.
\end{proof}

\begin{theorem}
Given a simple polygon $P$ with $n$ vertices and a fixed point $q \in P$, we build data structures of size $O(n)$ in $O(n)$ time to support the following.
Let $P'$ be the current simple polygon defined with $n'$ vertices.
Let $VP(q)$ be the visibility polygon of $q$ in $P'$. 
For a vertex $v$ inserted to (or, deleted from) $P'$, updating $VP(q)$ takes $O((k+1)(\lg{n'})^2)$ time, where $k$ is the number of updates required to the visibility polygon of $q$ due to the insertion (resp. deletion) of $v$ to (resp. from) $P'$.
\end{theorem}
\begin{proof}
The correctness follows from Lemmas~\ref{lem:2b} and \ref{lem:2c}.
The time and space complexities are given immediately after describing the corresponding algorithms.
\end{proof}

\section{Querying for the visibility polygon of $q$ when $q \notin P'$}
\label{sect:maintvpsimppolyext}

In this Section, we devise an algorithm to compute the visibility polygon of any query point located in $\mathbb{R}^2$ as the simple polygon is updated with vertex insertions and deletions.

The algorithm presented in Section~\ref{sect:rayshootrayrotate} is extended to compute the visibility polygon of any query point $q$ belonging to the current simple polygon $P'$.
Hence, the following corollary to Theorem~\ref{thm:visvertopencone}.

\begin{cor}
\label{cor:vpqptintdyn}
Our algorithm preprocesses the given simple polygon $P$ in $O(n)$ time and computes $O(n)$ spaced data structures to facilitate in answering visibility polygon $VP(q)$ of any given query point $q$ located interior to the current simple polygon $P'$ in $O(k(\lg{n'})^2)$ time amid vertex insertions and deletions. 
Here, $k$ is the number of vertices of $VP(q)$, $n$ is the number of vertices of $P$, and $n'$ is the number of vertices of $P'$.
\end{cor}

There are two cases in considering the visibility of a point $q \notin P'$: (i) $q$ lies outside of $CH(P')$, and (ii) $q$ belongs to $CH(P') \setminus P'$.
We consider these cases independently.
The following observations from \cite{books/visalgo/skghosh2007} are useful in reducing these problems to computing the visibility polygon of a point interior to a simple polygon:
\begin{itemize}
\item[*]
Let $q \notin CH(P')$ and let $t', t''$ be points of tangencies from $q$ to $CH(P')$. 
Also, let $\overrightarrow{qt'}$ occur when $\overrightarrow{qt''}$ is rotated in counterclockwise direction with center at $q$. 
Then the region bounded by line segments $qt', qt''$ and the vertices that occur while traversing $bd(P')$ in counterclockwise direction from $t'$ to $t''$ is a simple polygon $P''$.
(Refer to Fig.~\ref{fig:simppolyextinsi}.)
\item[*]
Let $q \in CH(P') \setminus P'$.
For some edge $e = t't''$ of $CH(P')$ with $t'$ occurring before $t''$ in the clockwise ordering of vertices of $bd(P')$, the region formed by $e$ together with the section of boundary of $P'$ from $t'$ to $t''$ in clockwise ordering of edges along $bd(P')$ is a simple polygon $P''$ (a pocket) that contains $q$.  
(Refer to Fig.~\ref{fig:simppolyextinsii}.)
\end{itemize}
We say the simple polygon $P''$ resultant in either of these cases as the simple polygon corresponding to points $t'$ and $t''$. 

The changes to algorithms devised in the last Section are mentioned herewith.
In $O((\lg{n})^2)$ time, using the dynamic planar point-location query \cite{journals/jal/GoodrichT97}, we determine whether $q \in P'$.
If $q \notin P'$ and points of tangency from $q$ to $CH(P')$ does exist (resp. does not exist), then $q \notin CH(P')$ (resp. $q \in CH(P') \setminus P'$).
If tangents exist from $q$ to $CH(P')$, then we compute the tangents from $q$ to $CH(P')$ in $O(\lg{n})$ time (refer \cite{books/compgeom/prep1985}).

\begin{figure}[h]
\centering
\includegraphics[totalheight=1.15in]{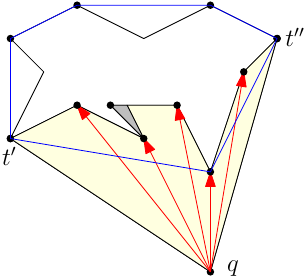}
\vspace{-0.15in}
\caption{\footnotesize Illustrating the exterior visibility of $P'$ from $q$ when $q \notin CH(P')$.}
\label{fig:simppolyextinsi}
\end{figure}

We dynamically maintain the convex hull $CH(P)$ of the current simple polygon $P'$ using the algorithm devised in Overmars et~al. \cite{journals/jcss/Overvan81}.
Whenever a vertex is inserted or deleted, we update the convex hull of the simple polygon.
To facilitate this, we preprocess $P$ defined with $n$ vertices to construct data structures of size $O(n)$ so that to update convex hull of the current simple polygon $P'$ in $O((\lg{n'})^2)$ time per vertex insertion or deletion using the algorithm from \cite{journals/jcss/Overvan81}.

Suppose that $q \notin CH(P')$. 
Let $t'$ and $t''$ be the points of tangency from $q$ to $CH(P')$.
Also, let $t'$ occur later than $t''$ in counterclockwise traversal of $bd(CH(P'))$.
We invoke $visvert$-$inopencone$ algorithm with $\overrightarrow{qt''}$ and $\overrightarrow{qt'}$ as the first and second parameters respectively.
Refer to Fig.~\ref{fig:simppolyextinsi}.
Let $P''$ be the simple polygon corresponding to points $t'$ and $t''$.
We do not explicitly compute $P''$ itself although the ray-shooting is limited to the interior of $P''$.
For every vertex $v'$ that is determined to be visible from $q$, we ray-shoot with ray $\overrightarrow{qv'}$ in $P''$ to determine the constructed edge that is incident to $v'$.

Suppose that $q \in CH(P') \setminus P'$.
First, we note that the ray-shooting algorithms given in \cite{journals/jal/GoodrichT97} work correctly within any one simple polygonal region of the planar subdivision.
And, same is the case with the ray-rotating algorithm from \cite{journals/comgeo/ChenW15a}.
To account for these constraints, as described below, we slightly modify the ray-rotating query algorithm invoked from $visvert$-$inopencone$.

\begin{figure}[h]
\centering
\includegraphics[totalheight=1.18in]{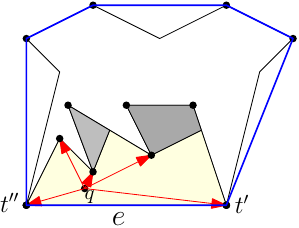}
\caption{\footnotesize Illustrating the exterior visibility of $P'$ from $q$ when $q \in CH(P') \setminus P'$.}
\label{fig:simppolyextinsii}
\end{figure}

During the invocation of the ray-shooting algorithm from the ray-rotating algorithm, if the ray-shooting algorithm determines that a ray $r$ with origin $q$ does not strike any point of $P'$, then in $O(\lg{n'})$ time we compute the edge $e$ of $CH(P')$ that gets struck by $r$.
The edge $e$ is found by searching in the dynamic hull tree corresponding to $CH(P')$ (Overmars et~al. \cite{journals/jcss/Overvan81}).
Let $t', t''$ be the endpoints of $e$.
Also, let $t'$ occurs before $t''$ in the clockwise ordering of vertices of $bd(P')$.
The closed region bounded by $e$ and the section of $bd(P')$ that occurs between $t'$ and $t''$ while traversing $bd(P')$ in clockwise direction starting from $t'$, is a simple polygon, say $P''$.
Refer to Fig.~\ref{fig:simppolyextinsii}.
Whenever we encounter an edge $e$ with these characteristics, we invoke the $visvert$-$inopencone$ algorithm with $qt'$ and $qt''$ as the first and second parameters respectively.
The rest of the ray-rotating algorithm from \cite{journals/comgeo/ChenW15a} is not changed.

The $visvert$-$inopencone$ algorithm performs $O(k)$ ray-rotation queries to output $k$ vertices that are visible from $q$. 
Since each ray-rotation query takes $O((\lg{n'})^2)$ time, the time complexity is $O(k(\lg{n'})^2)$.

\begin{lemma}
\label{lem:3c}
Let $P'$ be the current simple polygon.
Let $q$ be a point exterior to $P'$.
The visibility polygon $VP(q)$ of $q$ is updated correctly whenever a vertex $v$ is inserted to $P'$ or a vertex $v$ is deleted from $P'$.
\end{lemma}
\begin{proof}
When $q \notin CH(P')$, we find the points of tangency, $t'$ and $t''$, from $q$. 
Any vertex $x \in P'$ that does not lie in the $cone(\overrightarrow{qt''}, \overrightarrow{qt'})$ is not visible from $q$ as the line segment $qx$ is intersected by one of the edges of $P''$.
Hence, we invoke $visvert$-$inopencone$ algorithm in this cone to determine the vertices of $VP(q)$ that are visible from $q$. 
When $q \in CH(P') \setminus P'$, since $q \notin P'$, $q$ lies in one of the pockets formed by vertices of $P'$ and $CH(P')$. 
We identify the pocket $R$ in which $q$ resides by finding vertices $t'$ and $t''$.
This is accomplished by traversing the dynamic hull tree to find the edge $t't''$ of $CH(P') \setminus P'$ that belongs to pocket $R$.
Any vertex $x \in P'$ that does not belong to sequence of vertices that occur in traversing $bd(P')$ from $t'$ to $t''$ in clockwise direction is not visible from $q$ as the line segment $qx$ is intersected by an edge of $P'$. 
Hence, it suffices to invoke $visvert$-$inopencone$ algorithm with respect to the $opencone(qt', qt'')$. 
Further, for every updated vertex, its corresponding constructed vertex is computed correctly with a ray-shooting query.
\end{proof}

\begin{lemma}
\label{lem:vpqptextdyn}
By preprocessing the given initial simple polygon $P$ defined with $n$ vertices, data structures of size $O(n)$ are computed. 
These facilitate in adding or deleting any vertex from the current simple polygon $P'$ in $O((\lg{n'})^2)$ time, and to output the visibility polygon $VP(q)$ of any query point $q$ located exterior to $P'$ in $O(k(\lg{n'})^2)$ time.
Here, $k$ is the number of vertices of $VP(q)$, and $n'$ is the number of vertices of $P'$. 
\end{lemma}
\begin{proof}
The correctness follows from Lemma~\ref{lem:3c}.
The preprocessing time of dynamic hull algorithm of Overmars et~al. \cite{journals/jcss/Overvan81} is $O(n)$ and it takes $O((\lg{n'})^2)$ time to update the hull tree per an insertion or deletion of a vertex of the current simple polygon $P'$.
Further, it takes $O((\lg{n'})^2)$ time to search in the current hull tree to find the edge $t't''$ of $CH(P')$.
The rest of the time and space complexities are argued above.
\end{proof}

\noindent
From Corollary~\ref{cor:vpqptintdyn} and Lemma~\ref{lem:vpqptextdyn}, the following theorem is immediate.

\begin{theorem}
By preprocessing the given initial simple polygon $P$ defined with $n$ vertices, data structures of size $O(n)$ are computed. 
These facilitate in adding or deleting any vertex from the current simple polygon $P'$ in $O((\lg{n'})^2)$ time, and to output the visibility polygon $VP(q)$ of any query point $q \in \mathbb{R}^2$ in $O(k(\lg{n'})^2)$ time.
Here, $k$ is the number of vertices of $VP(q)$, and $n'$ is the number of vertices of $P'$. 
\end{theorem}

\section{Maintaining the weak visibility polygon of a fixed line segment amid vertex insertions}
\label{sect:wvpsimppolyincr}

In this section, we devise an algorithm to maintain the weak visibility polygon of a line segment located interior to the given simple polygon, as the vertices are added to the simple polygon.

When $pq$ is a line segment contained within a simple polygon $Q$, we can partition $Q$ into two simple polygons $Q_1$ and $Q_2$ by extending segment $pq$ on both sides until it hits $bd(Q)$. 
Then, the $WVP$ of $pq$ in $Q$ is the union of the $WVP$ of $pq$ in $Q_1$ and the $WVP$ of $pq$ in $Q_2$. 
To update $WVP$ of $pq$ in $Q$, we need to update the $WVP$s of $pq$ in both the simple polygons $Q_1$ and $Q_2$.
So to maintain the $WVP$ of a line segment in a simple polygon $Q$ is reduced to computing the $WVP$s of an edge of two simple polygons $Q_1$ and $Q_2$.
Hence, in this Section, we consider the problem of updating the weak visibility polygon $WVP(pq)$ of an edge $pq$ of a simple polygon.
With a slight abuse of notation, we denote $WVP(pq)$ with $WVP$ when $pq$ is clear from the context.

We review the algorithm by Guibas et~al. \cite{journals/algorithmica/GuibasHLST87}, which is used in computing the initial $WVP(pq)$ in simple polygon $P$.
First, the \emph{shortest path tree} rooted at $p$, denoted by $SPT(p)$, is computed: note that $SPT(p)$ is the union of shortest path from $p$ to every vertex $v_i \in P$.
Then the $SPT(q)$ is computed. 
As part of depth first traversal of $SPT(p)$ (resp. $SPT(q)$), if the shortest path to a child $v_j$ of $v_i$ makes a right turn (resp. left turn) at $v_i$, then a line segment is constructed by extending $v_kv_i$ to intersect $bd(P)$, where $v_k$ is the parent of $v_i$ in $SPT(p)$ (resp. $SPT(q)$).
The portion of $P$ lying on the right side (resp. left side) of line segments such as $v_kv_i$ in $SPT(p)$ (resp. $SPT(q)$) do not belong to $WVP(pq)$.
It was shown in \cite{journals/algorithmica/GuibasHLST87} that every point belonging to $P$ but not belonging to pruned regions of $P$ is weakly visible from $pq$.
From the algorithm of Guibas et~al.~\cite{journals/algorithmica/GuibasHLST87}, as a new vertex is added to $P$, it is apparent that to update the $WVP(pq)$, we need to update $SPT(p)$ and $SPT(q)$.
This is to determine the vertices of $P$ that do not belong to $WVP(pq)$ using these shortest path trees. 

For updating $SPT(p)$ (as well as $SPT(q)$) with the insertion of new vertices to $P$, we use the algorithm by Kapoor and Singh \cite{conf/fsttcs/KapoorS96}.
Their algorithm to update $SPT(p)$ is divided into three phases: the first phase computes every line segment $e_i$ of $SPT(p)$ that intersects with $\triangle vv_iv_{i+1}$; in the second phase, for each such $e_i$, shortest path tree $SPT(p)$ is updated to include the endpoint $z_i$ of $e_i$;  and in the final phase, the algorithm updates $SPT(p)$ to include every child of all such $z_i$. 
As a result, when a vertex $v$ is inserted to current simple polygon $P'$ defined with $n'$ vertices, the algorithm in \cite{conf/fsttcs/KapoorS96} updates $SPT(p)$ in $O(k\lg{(\frac{n'}{k})})$ time in the worst-case.
Here, $k$ is the number of changes made to $SPT(p)$ due to the insertion of $v$ to $P'$.

A \emph{funnel} consists of a vertex $r$, called the \emph{root}, and a line segment $xy$, called the \emph{base} of the funnel.
The \emph{sides} of the funnel are $SP(r,x)$ and $SP(r,y)$.
In \cite{conf/fsttcs/KapoorS96}, $SPT(p)$ is stored as a set of funnels.
For a funnel $F$ with root $r$ and base $xy$, the line segments from $r$ to $x$ and from $r$ to $y$ are stored in balanced binary search trees $T_1$ and $T_2$ respectively. 
The root nodes of $T_1$ and $T_2$ have pointers between them. 
It can be seen that every edge $e$ that is not a boundary edge lies in exactly two funnels say $F_i$ and $F_j$. 
The node in $F_i$ that corresponds to $e$ has a pointer to the node in $F_j$ that corresponds to $e$, and vice versa.
The $SPT(p)$ stored as funnels takes $O(n)$ space, where $n$ is the number of vertices of a simple polygon.

Similar to Section~\ref{sect:maintvpsimppolyint}, vertices of $P$ and $WVP$ are stored in balanced binary search trees $B_P$ and $B_{WVP}$.
Additionally, nodes in $B_{P}$ and $B_{WVP}$ which correspond to the same vertex $u_i \in WVP$ contain pointers to refer between themselves. 
For each region $R_i$, vertices are stored in a balanced binary search tree $T_{R_i}$.
And, for every node $v_j \in R_i$, node in $T_{R_i}$ (resp. $B_P$) that correspond to $v_j$ contain a pointer to a node in $B_P$ (resp. $T_{R_i}$) that correspond to $v_j$.
We call the pointer from a node in $B_P$ to a node in $T_{R_i}$ a {\it tree pointer}.
For any vertex $v_i \notin WVP$, the tree pointer of $v_i$ is set to null.
All of these data structures, together with the ones to store funnels, take $O(n)$ space.

In our incremental algorithm, whenever a vertex is inserted, we first update $SPT(p)$ and $SPT(q)$.
Then, with the depth-first traversal starting from $v$ in the updated $SPT(p)$ and $SPT(q)$, we remove the regions that are not entirely visible from $pq$. 
The tree pointers are used to identify which of the cases (and sub-cases) are applicable in a given context while the algorithm splits or joins regions' as it proceeds. 

Since we maintain the $WVP$ of edge $pq$, we assume that no vertex is inserted between vertices $p$ and $q$.

\vspace{-0.1in}

\subsubsection*{Preprocessing\\}

As mentioned above, we compute $SPT(p)$ and $SPT(q)$ using \cite{journals/algorithmica/GuibasHLST87} and store them as funnels using the algorithm from \cite{conf/fsttcs/KapoorS96}. 
We also compute $WVP$ and the balanced binary search trees for each region.
The preprocessing phase takes $O(n)$ time in the worst-case.

\vspace{-0.1in}

\subsubsection*{Updating $WVP(pq)$\\}

Let $P'$ be the current simple polygon.
Let $v$ be the vertex inserted between vertices $v_i$ and $v_{i+1}$ of $P'$.
Also, let $pq$ be the edge whose $WVP$ is being maintained. 

\begin{figure}[h]
\centering
\includegraphics[totalheight=1.1in]{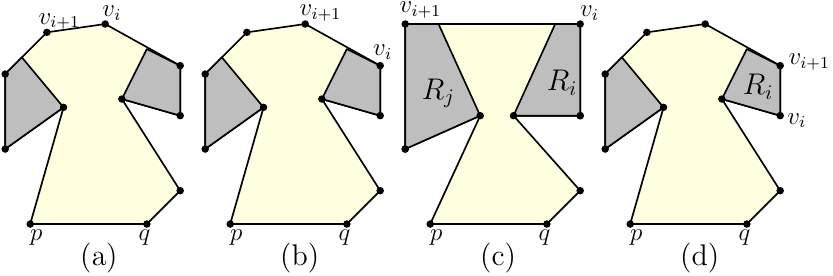}
\vspace{-0.1in}
\caption{Illustrating cases based upon the locations of $v_i$ and $v_{i+1}$.}
\label{fig:11}
\end{figure}

We have the following four cases depending on the relative location of $v_i$ and $v_{i+1}$ with respective to any two disjoint occluded regions, say $R_i$ and $R_j$, and the weak visibility polygon $WVP$ of $pq$:
\begin{enumerate}[(a)]
\item $v_i, v_{i+1} \in WVP$ \hspace{0.03in} (Fig.~\ref{fig:11}(a))
\item $v_i \notin WVP$ and $v_{i+1} \in WVP$, or vice versa \hspace{0.03in} (Fig.~\ref{fig:11}(b))
\item $v_i,v_{i+1} \notin WVP$ and $v_i\in R_i, v_{i+1}\in R_j$ \hspace{0.03in} (Fig.~\ref{fig:11}(c))
\item $v_i,v_{i+1} \notin WVP$ and both $v_i, v_{i+1}\in R_i$ \hspace{0.03in} (Fig.~\ref{fig:11}(d))
\end{enumerate}

Further, depending on the location of the newly inserted vertex $v$, there are two sub-cases in each of these cases: (I) $v \in P'$, (II) $v \notin P'$.
We first consider the cases (a)-(d) when $v \in P'$.
We make use of join, split and insert operations of the balanced binary search trees described in Tarjan \cite{books/dsnetworkalgo/tarjan1983}, each of these operations takes $O(\lg{n})$ time.
Identifying these cases is explained later.
We use $k_1$ and $k_2$ to denote the number of changes required to update $SPT(p)$ and $SPT(q)$ respectively with the insertion of new vertex $v$.\\

\begin{figure}[h]
\centering
\includegraphics[totalheight=1.3in]{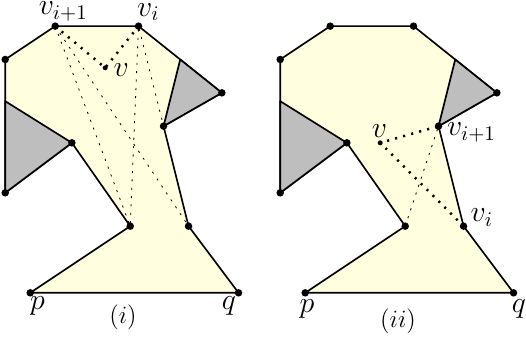}
\vspace{-0.1in}
\caption{Case (a): (i) $SP(p,v_i)$ and $SP(p,v_{i+1})$ do not intersect $\triangle vv_iv_{i+1}$ and (ii) $SP(p,v_{i+1})$ intersects $\triangle vv_iv_{i+1}$}
\label{fig:12}
\end{figure}

\noindent
{\bf Case (a)} \\ 
(i) When $SP(p,v_i)$ and $SP(p,v_{i+1})$ do not intersect $\triangle vv_iv_{i+1}$, (Fig. \ref{fig:12}(i)), we add $v$ to $SPT(p)$ by finding a tangent from $v$ to funnel with base $v_iv_{i+1}$ in $SPT(p)$. 
Analogously, if $SP(q,v_i)$ and $SP(q,v_{i+1})$ do not intersect $\triangle vv_iv_{i+1}$, then $v$ is added to $SPT(q)$ as well. 
And, $v$ is added to $WVP(pq)$.
No further changes are required.
(ii) We now consider the case when $SP(p,v_{i+1})$ intersects $\triangle vv_iv_{i+1}$ (Fig. \ref{fig:12}(ii)).
The handling of this case is detailed below. 

\begin{figure}[h]
\centering
\includegraphics[totalheight=1.2in]{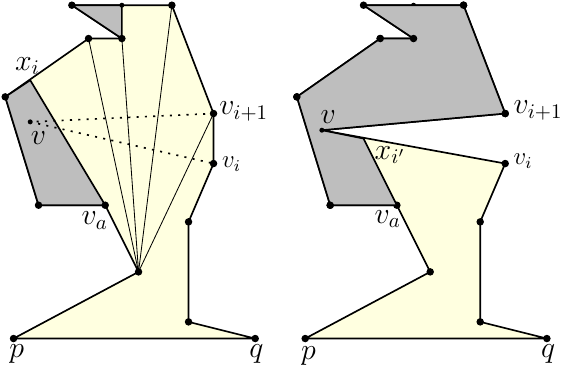}
\vspace{-0.1in}
\caption{Illustrating Case (a)(ii)(1): $v \in R_i$}
\label{fig:13}
\end{figure}

\noindent
Case (1):
Suppose that the vertex $v$ lies in a region $R_i$ with associated constructed edge $v_ax_i$ (Fig. \ref{fig:13}). 
While updating $SPT(p)$ using \cite{conf/fsttcs/KapoorS96}, we check if one of the edges in $SPT(p)$ intersecting $\triangle vv_iv_{i+1}$ is $v_ax_i$.
If it is, all the vertices on $bd(P'')$ from $v_{i+1}$ to $x_i$ are joined into a single occluded region and the vertices on $bd(WVP)$ from $v_{i+1}$ to $x_i$ are deleted from $WVP$. 
This is because $\triangle vv_iv_{i+1}$ blocks the visibility of these vertices from $pq$.
And, this is accomplished with the union of all the regions $R_j$ that are encountered while traversing along $bd(WVP)$ from $v_{i+1}$. 
If there are $O(k_1)$ such regions, then joining the new region with $R_i$ takes $O(k_1\lg{n})$ time. 
The point of intersection $x_{i'}$ of $vv_i$ and $v_ax_i$ is added as a constructed vertex to both $WVP$ and $SPT(p)$, and $R_i$ is associated with $x_{i'}$. 
Since $v$ is located in an occluded region $R_i$, from the algorithm in \cite{journals/cvgip/LeeL86}, we know that $v_ax_i$ belongs to either $SPT(p)$ or $SPT(q)$ of $P'$.
Hence, if $v_ax_i \notin SPT(p)$ then $v_ax_i \in SPT(q)$.
In the latter case, we do the analogous updates to $WVP$ while updating $SPT(q)$. \\

\begin{figure}[h]
\centering
\includegraphics[totalheight=1.3in]{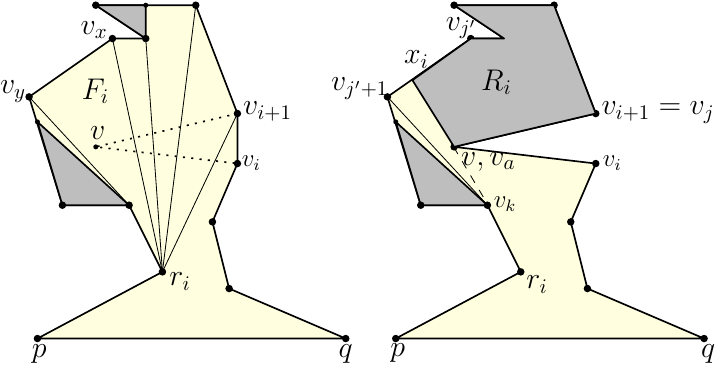}
\vspace{-0.1in}
\caption{Illustrating Case (a)(ii)(2): $v \in WVP$}
\label{fig:14}
\end{figure}

\noindent
Case (2):
If $v \in WVP$, then no constructed edge intersects $\triangle vv_iv_{i+1}$ (refer to Fig.~\ref{fig:14}). 
We start with updating $SPT(p)$. 
Let $F_i$ be the funnel containing $v$ with root $r_i$ and base $v_xv_y$.
We then perform the depth first traversal in $SPT(p)$, starting from vertex $r_i$ while traversing those simple paths in $SPT(p)$ that intersect $\triangle vv_iv_{i+1}$. 
For any vertex $v_j$ with parent $v_a$ in $SPT(p)$, if $SP(r_i, v_j)$ takes a right turn at vertex $v_a$ then  as described herewith we add a constructed edge $v_ax_i$.
This is accomplished by finding the leftmost child $v_{j'}$ of $v_a$ for which $SP(v,v_{j'})$ takes a right turn at $v_a$.
Let $v_k$ be the parent of $v_a$ in $SPT(p)$, then the constructed vertex $x_i$ is found by extending $v_kv_a$ on to $v_{j'}v_{j'+1}$. 
This results in a new region $R_i$.
(Refer to right subfigure of Fig.~\ref{fig:14}.)
We add to $R_i$ every vertex that occur in traversing $bd(WVP)$ in counterclockwise direction from vertex $v_{j} (= v_{i+1})$ to vertex $v_a$ and delete the same set of vertices from $WVP$.
As part of this, every region $R_j$ that is encountered in this traversal is joined with $R_i$. 
If there are $O(k_1)$ such regions, this step takes $O(k_1\lg{n})$ time.
Finally, the constructed vertex $x_i$ is added to $WVP$ and $SPT(p)$.
To update $SPT(q)$, we find every edge $e_i$ in $SPT(q)$ that intersect $\triangle vv_iv_{i+1}$. 
If the hidden endpoint $z_i$ of an edge $e_i$ has a non-null tree pointer (it lies in some region $R_j$), then $z_i$ is deleted along with its children. 
Similar to the depth-first traversal in $SPT(p)$ above, we perform the depth-first traversal of $SPT(q)$ while considering subtrees that contain vertices not determined to be not weakly visible from $pq$ during the depth-first traversal of $SPT(q)$.
This completes the update procedure.
Note that whole of this algorithm takes $O(k_2\lg{n})$ time in the worst-case, where $k_2$ is the number of changes required to update $SPT(q)$ with the insertion of vertex $v$.
When $SP(q,v_i)$ intersects $\triangle vv_iv_{i+1}$, we first update $SPT(q)$ and steps analogous to above are performed. 
Including updating both of $SPT(p)$ and $SPT(q)$, our algorithm takes $O((k_1+k_2)\lg{n})$ time in the worst-case.\\

\noindent
{\bf Case (b)} \\
We assume that $v_i \notin WVP$ and $v_{i+1} \in WVP$. Let $v_i \in R_i$ with the constructed vertex $x_i$ and the constructed edge $v_ax_i$. 
If $vv_{i+1}$ intersects $v_ax_i$ at $x_{i'}$ (Fig. \ref{fig:15}(i)), 
then $v \in R_i$.

\begin{figure}[h]
\centering
\includegraphics[totalheight=1.3in]{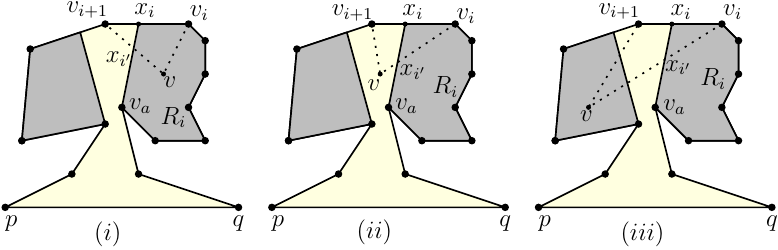}
\vspace{-0.1in}
\caption{Illustrating Case (b): (i) $v \in R_i$; (ii) $v \in WVP$; (iii) $v \notin R_i \cup WVP$}
\label{fig:15}
\end{figure}

We update $WVP$ by deleting $x_i$ and adding $x_{i'}$. 
Similarly, we update $SPT(p)$ and $SPT(q)$ by deleting $x_i$ and adding $x_{i'}$. 
Finally, we add $v$ to $R_i$. 
Overall, including updating $SPT(p)$ and $SPT(q)$, handling this sub-case takes $O(\lg{n})$ time.

When $v \notin R_i$ (Figs. \ref{fig:15}(ii), \ref{fig:15}(iii)), we find the point of intersection $x_{i'}$ of line segments $vv_i$ and $v_ax_i$.
We update $WVP$ by deleting $x_i$ and adding $x_{i'}$.
Similarly, we update $SPT(p)$ and $SPT(q)$ by deleting $x_i$ and adding $x_{i'}$.
This case is then similar to case (a) where $v_i$ is replaced with $x_{i'}$; 
hence, takes $O((k_1+k_2)\lg{n})$ time.\\

\begin{figure}[h]
\centering
\includegraphics[totalheight=1.3in]{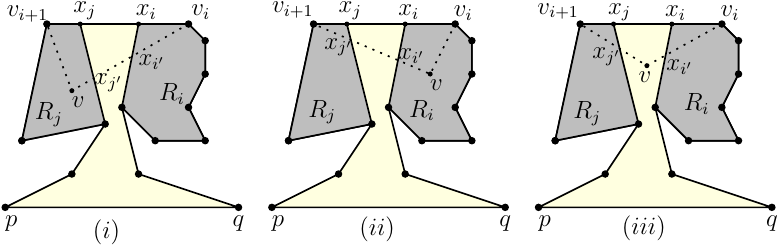}
\vspace{-0.1in}
\caption{Illustrating Case (c): (i) $v \in R_j$, (ii) $v\in R_i$, and (iii) $v \in WVP$}
\label{fig:16}
\end{figure}

\noindent
{\bf Case (c)} \\
Let $v_i \in R_i$ and let $v_{i+1} \in R_j$.
Also, let $x_i$ and $x_j$ be the respective constructed vertices on $v_iv_{i+1}$.
Then we have three sub-cases, which are shown in Fig.~\ref{fig:16}.
We note that there are no other sub-cases, as $v$ cannot be in any occluded region other than $R_i$ and $R_j$. 
(Whenever this is the case, either the line segment $vv_i$ or the line segment $vv_{i+1}$ intersects with the boundary of the simple polygon.)
Determining and handling any of these cases is similar to case (b).
However, here we additionally update $SPT(p)$ and $SPT(q)$ by adding the new constructed vertices (intersection points $x_{i'}$ and $x_{j'}$ of $\triangle vv_iv_{i+1}$ with the constructed edges).
This additional step takes $O(\lg{n})$ time.\\

\noindent
{\bf Case (d)} \\
Now we have $v_i, v_{i+1} \in R_i$.
Let $x_i$ be the associated constructed vertex of $R_i$ and $v_ax_i$ be the corresponding constructed edge.
If $v \in R_i$ (Fig. \ref{fig:17}(i)) then no change is made to $WVP$ but we need to add $v$ to $T_{R_i}$  which takes $O(\lg{n})$ time.

\begin{figure}[h]
\centering
\includegraphics[totalheight=1.3in]{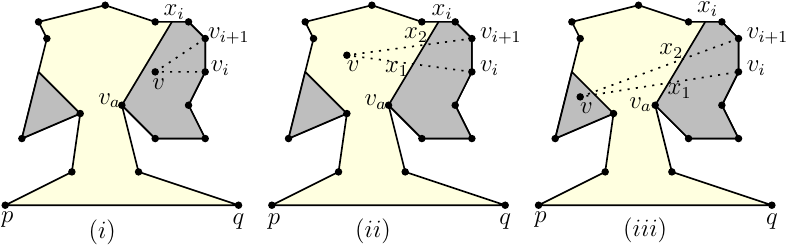}
\vspace{-0.1in}
\caption{Illustrating Case (d): (i) $v \in R_i$, (ii) $v \in WVP$, and (iii) $v \notin R_i \cup WVP$}
\label{fig:17}
\end{figure}

If $v \notin R_i$ (Figs. \ref{fig:17}(ii), \ref{fig:17}(iii)), 
we assume a portion of $vv_i$ is visible from $pq$. 
(When a portion of $vv_{i+1}$ is visible from $pq$, it is handled analogously.)
Then the polygonal region $R_i$ is split into two regions.
\begin{wrapfigure}{r}{0.26\textwidth}
\centering
\begin{minipage}[t]{\linewidth}
\centering
\includegraphics[totalheight=1.3in]{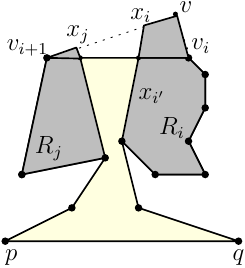}
\vspace{-0.1in}
\caption{Illustrating Case (II): when $v \notin P'$}
\label{fig:17c}
\end{minipage}
\end{wrapfigure}
Hence, we split the tree $T_{R_i}$ as well. 
This results in two balanced binary search trees:
tree $T_{R_{i1}}$ comprising vertices that occur in traversing $bd(P'')$ from $v_i$ to $v_a$ in clockwise order;
and, tree $T_{R_{i2}}$ comprising vertices that occur in traversing $bd(P'')$ from $v_{i+1}$ to $x_i$ in counterclockwise order.
Let $vv_i$ and $vv_{i+1}$ respectively intersect $v_ax_i$ at points $x_1$ and $x_2$ respectively. 
We add $x_1$ and $x_2$ to $SPT(p)$ and $SPT(q)$.
Also, we associate $x_1$ as a constructed vertex with $T_{R_{i1}}$.
This case is then similar to case (a) with $\triangle vx_1x_2$ except that in the last step, the tree $T_{R_{i2}}$ is joined with a newly formed tree.  
This completes the case (d).
And, this case require $O((k_1+k_2)\lg{n})$ time.

In Case (II), i.e., when $v \notin P'$, the resultant cases are each of the four cases (a) to (d) can be handled analogously. 
If one/two sections of $v_iv_{i+1}$ are not weakly visible from $pq$, then we extend the respective constructed edges that are incident to $v_iv_{i+1}$ into $\triangle$ $vv_iv_{i+1}$.
(Refer to Fig.~\ref{fig:17c}.)
We add the vertex $v$ to $SPT(p)$ and $SPT(q)$.
Further, the tree comprising vertices that occur in traversing the boundary of each modified occluded region is updated by adding the constructed vertices and possibly vertex $v$. 
And, each of these cases again take $O((k_1+k_2)\lg{n})$ time in the worst-case.
\\

\noindent {\bf Distinguishing cases (a)-(d)}  \\
To determine the appropriate case among cases (a)-(d), we follow the tree pointers associated with $v_i$ and $v_{i+1}$ in $L_P$. 
For $v_i$, if the tree pointer is not null, then it points to a node in a tree, say $T_{R_b}$. 
We follow the parent pointers from this node to reach the root of $T_{R_b}$, say $r_1$, which takes $O(\lg{n})$ time. 
Similarly, $r_2$ is found if the tree pointer for $v_{i+1}$ is not null. 
The appropriate case can thus be found using these tree pointers.
Based on the angle between $v_iv$ and $vv_{i+1}$, we determine whether $v \in P'$ or $v \notin P'$.

\begin{lemma}
\label{lem:wvpincr1}
The cases (a)-(d) are exhaustive. 
Further, the corresponding sub-cases are exhaustive as well.
\end{lemma} 
\noindent
\begin{proof}
The cases (a)-(d) are defined based on the occluded regions of $P'$ to which vertices $v_i$ and $v_{i+1}$ belong to: 
(1) both $v_i$ and $v_{i+1}$ belong to WVP, (2) only one of them belongs to WVP, and (3) neither of $v_i, v_{i+1}$ belongs to WVP.
Note that (1) is same as the case (a) in our algorithm and (2) is same as the case (b).
Since there are multiple occluded regions, case (3) can further be divided into two sub-cases: (i)$v_i$ and $v_{i+1}$ belong to two distinct occluded regions, (ii) both $v_i$ and $v_{i+1}$ belong to the same region.
Note that (3)(i) is same as the case (c) of our algorithm and (3)(ii) is same as the case (d).

The sub-cases of case (a) are exhaustive and are considered in the algorithm:
$SP(p, v_i)$ (resp. $SP(p, v_{i+1})$, $SP(q, v_i)$, $SP(q, v_{i+1})$) do not intersect $\triangle vv_iv_{i+1}$; $SP(p, v_i)$ (resp. $SP(p, v_{i+1})$, $SP(q, v_i)$, $SP(q, v_{i+1})$) intersects $\triangle vv_iv_{i+1}$ and $v$ is weakly visible to $pq$; $SP(p, v_i)$ intersects $\triangle vv_iv_{i+1}$ but $v$ is not weakly visible to $pq$.
All possible sub-cases of case (b) are considered: the sub-cases in which $v$ belongs to same region as $v_i$, vertex $v$ belonging to a different region from the region to which $v_i$ belongs to, and $v$ being located in WVP are considered. 
The sub-cases of case (c) consider $v$ being located in the occluded region $R_i$ to which $v_i$ belongs, $v$ being located in the occluded region $R_j$ to which $v_{i+1}$ belongs, and $v$ belonging to the $WVP(pq)$.
Further, the sub-cases of case (d) consider $v$ being located in the same region $R_i$ as vertices $v_i$ and $v_{i+1}$, vertex $v$ located in $WVP$, and $v$ located in a distinct occluded region from $R_i$.
\end{proof}

\begin{lemma}
\label{lem:wvpincr2}
After inserting a vertex $v$, a vertex $u_i \in WVP(pq)$ is deleted from $WVP(pq)$ if and only if it is not weakly visible from $pq$ and a (constructed) vertex is added to $WVP(pq)$ whenever it is weakly visible from $pq$.
\end{lemma}
\noindent
\begin{proof}
If a vertex $u_i \in WVP$ is not visible from $p$, then either $SP(p,u_i)$ takes a right turn at some vertex $v_j$ or $SP(q,u_i)$ takes a left turn at some vertex $v_{j'}$. 
But such vertices are removed from $WVP$ during the depth-first traversals of $SPT(p)$ and $SPT(q)$.
Similarly, if a vertex $u_i \in WVP$ is deleted from $WVP$, then either $SP(p,u_i)$ turns to the right or $SP(q,u_i)$ turns to the left.
This implies $u_i$ is not visible from $pq$.
These invariants are ensured in all the cases and sub-cases mentioned in the algorithm.
The same is applicable to constructed vertices as well: a constructed vertex $v_j$ is added in any sub-case if and only if $SP(p, v_j)$ does not make a right turn at any intermediate vertex along $SP(p, v_j)$ and $SP(q, v_j)$ does not take a left turn at any intermediate vertex along $SP(q, v_j)$.
\end{proof}

\begin{theorem}
When a new vertex $v$ is inserted to current simple polygon $P'$, the incremental algorithm to update the weak visibility polygon of an edge of $P'$ works correctly.
After preprocessing the initial simple polygon $P$ and the edge $pq$ of $P$ to build data structures of size $O(n)$ in $O(n)$ time, the weak visibility polygon of $pq$ in the current simple polygon $P'$ is updated in $O((k+1)\lg{n'})$ time whenever a vertex $v$ is inserted to $P'$. 
Here, $n$ is the number of vertices of $P$, $n'$ is the number of vertices of $P'$, and $k$ is the total number of changes required to update $SPT(p)$ and $SPT(q)$ due to the insertion of vertex $v$ to $P'$.
\end{theorem}
\begin{proof}
Correctness follows from Lemmas~\ref{lem:wvpincr1} and \ref{lem:wvpincr2}.
Computing $SPT(p), SPT(q)$ and storing them as funnels together with computing $WVP$ of $pq$ in $P$ together with building balanced binary search trees corresponding to occluded regions with respect to $WVP$ is done in $O(n)$ time. 
Identifying the appropriate case requires $O(\lg{n'})$ time.
As in the algorithm for updating the $WVP$ of a line segment, each case takes $O((k_1 + k_2) \lg{n'})$ time. 
Since $k$ equals to zero in sub-case (i) of case (d), overall time complexity for updating the weak visibility polygon is $O((k+1)\lg{n'})$.

\end{proof}

\begin{cor}
\label{corr:wvplineseg}
When a new vertex $v$ is inserted to current simple polygon $P'$, the incremental algorithm to update the weak visibility polygon of a fixed line segment $pq$ located interior to $P'$ works correctly.
After preprocessing the initial simple polygon $P$ and the fixed line segment $pq$ to build data structures of size $O(n)$ in $O(n)$ time, the weak visibility polygon of $pq$ in the current simple polygon $P'$ is updated in $O((k+1)\lg{n'})$ time whenever a vertex $v$ is inserted to $P'$.
Here, $n$ is the number of vertices of $P$, $n'$ is the number of vertices of $P'$, and $k$ is the total number of changes required to update $SPT(p)$ and $SPT(q)$ due to the insertion of vertex $v$ to $P'$.
\end{cor}

\section{Querying for the weak visibility polygon of $l$ when $l \in P'$}
\label{sect:wvpsimppolyint}

Let $s$ be a query line segment interior to the given simple polygon $P$.
Let $a$ and $b$ be the endpoints of $s$.
The weak visibility polygon $WVP(s)$ of $s$ is $\bigcup_{p \in s} VP(p)$.
The algorithm needs to capture the combinatorial representation changes of $VP(p)$ as the point $p$ moves from $a$ to $b$ along $s$.
First, we describe the notation and algorithm for computing the weak visibility polygon from \cite{journals/comgeo/ChenW15a}.
For any two vertices $v', v'' \in P$, let $b'$ be the point of intersection of ray $v''v'$ with the boundary of $P$.
If $b' \ne v'$, then the line segment $b'v'$ is termed as a {\it critical constraint} of $P$.
(Refer to Fig.~$2$ in \cite{journals/comgeo/ChenW15a}.)
Initially, point $p$ is at $a$ and $VP(p)$ is same as $VP(a)$.
As $p$ moves from $a$ to $b$ along $s$, a new vertex of $P$ could be added to the weak visibility polygon of $s$ whenever $p$ crosses a  critical constraint of $P$ \cite{journals/dcg/AronovGTZ02,journals/comgeo/BoseLM02}.
We also need the principal child definition from \cite{journals/comgeo/ChenW15a}, which is mentioned herewith. 
The {\it shortest path tree} rooted at $p$, $SPT(p)$, is the union of the shortest paths in $P$ from $p$ to all vertices of $P$.
A vertex of $P$ is in $VP(p)$ if and only if it is a child of $p$ in $SPT(p)$. 
For any child $v$ of $p$ in the tree $SPT(p)$, the {\it principal child} of $v$ is the child $w$ of $v$ in $SPT(p)$ such that the angle between rays $vw$ and $pv$ is smallest as compared with the angle between rays $vw'$ and $pv$ for any other child $w' \ne w$ of $v$.
(Refer to Fig.~$3$ in \cite{journals/comgeo/ChenW15a}.)
As $p$ moves along $s$, the next critical constraint that it encounters is characterized in the following observation from \cite{journals/dcg/AronovGTZ02}:

\begin{lemma}[from \cite{journals/dcg/AronovGTZ02}] 
\label{lem:critconstr}
The next critical constraint of a point $p$ is defined by two vertices of $P$ that are either two consecutive children of $SPT(p)$ or one, say $v$, being a child of $p$ and the other being the principal child of $v$.
\end{lemma}

We preprocess the same data structures as in the case of dynamic algorithms for updating visibility polygon of a point interior to the given initial simple polygon $P$ (Section~\ref{sect:maintvpsimppolyint}).
Let $P'$ be the current simple polygon.
Also, let $n'$ be the number of vertices of $P'$.
Given a line segment $s$ with endpoints $a$ and $b$, we first compute the visibility polygon when $p$ is at point $a$.
Let $r_1$ and $r_2$ be two arbitrary rays whose origin is at $a$.
The $VP(a)$ is computed in $O(k(\lg{n'})^2)$ time by invoking the $visvert$-$inopencone$ algorithm with $r_1$ as the first parameter as well as the second parameter.
For any vertex $v$ visible to $p$, as described in \cite{journals/comgeo/ChenW15a}, with one ray-rotating query, we determine the principal child of $v$ in $O((\lg{n})^2)$ time.
The critical constraints that intersect $s$ are stored in a priority queue $Q$, with the key value of a critical constraint $c$ equal to the distance of the point of intersection of $c$ and $s$ from $a$.
The extract minimum on $Q$ determines the next critical constraint that $p$ strikes.
After crossing a critical constraint, if $p$ sees an additional vertex $v'$ of $P'$, then we insert $v'$ into the appropriate position in $L_{vp}$.
For each critical constaint $c'$ that arise due to $v'$ which intersects with $s$, $c'$ is pushed into $Q$ with the distance from $a$ to the point of intersection of $c'$ and $s$ as the key.
To avoid updates to keys in $Q$, even though $p$ is moving along $ab$, every key in $Q$ represents the distance between $a$ and the point of intersection of the corresponding critical constraint and $s$.
As and when a vertex of $P'$ is determined to be weakly visible from $s$, we compute the constructed edge with a ray-shooting query from the point of intersection of the corresponding critical constraint with the $bd(P)$.
When $p$ reaches point $b$ (endpoint of $s$), $L_{vp}$ is updated so that it represents the $WVP(s)$.

\begin{lemma}
\label{lem:4t}
The time taken to query for the weak visibility polygon is $O(k(\lg{n'})^2)$ time, where $k$ is the output complexity and $n'$ is the number of vertices of current simple polygon $P'$.
\end{lemma}
\begin{proof}
Due to algorithm from Section~\ref{sect:maintvpsimppolyint}, determining the $VP(a)$ in $P'$ takes $O(k_1 (\lg{n'})^2)$, where $k_1$ is the number of vertices of $VP(a)$. 
Each of these $k_1$ vertices are added to $WVP(ab)$. 
Further, for every critical constraint that is computed, at most one vertex is added to $WVP(ab)$. 
If there are $O(k_2)$ critical constraints, this part of the algorithm takes $O(k_2 (\lg{n'})^2)$ time.
Noting that $k = k_1+k_2$, the stated time complexity includes operations associated with priority queue $Q$ as well.
\end{proof}

\begin{lemma}
\label{lem:4c}
A vertex is added to WVP(pq) if and only if it is visible from at least one point on line segment $pq$.
\end{lemma}
\begin{proof}
Immediate from \cite{journals/comgeo/ChenW15a}.
\end{proof}

\begin{theorem}
With $O(n)$ time preprocessing of the initial simple polygon $P$, data structures of size $O(n)$ are computed to facilitate vertex insertion, and vertex deletion in $O((\lg{n'})^2)$ time, and to output the weak visibility polygon of a query line segment located interior to the current simple polygon $P'$ in $O(k(\lg{n'})^2)$ time. 
Here, $k$ is the output complexity, $n$ is the number of vertices of $P$, and $n'$ is the number of vertices of $P'$. 
\end{theorem}
\begin{proof}
Correctness follows from Lemmas~\ref{lem:4t} and \ref{lem:4c}.
The preprocessing time is same as for the visibility polygon maintenance algorithm given in Section~\ref{sect:maintvpsimppolyint}.
\end{proof}

\section{Conclusions}
\label{sect:conclu}

We have presented algorithms to dynamically maintain as well as to query for the visibility polygon of a point $q$ located in the simple polygon as that simple polygon is updated.
For any point $q$ exterior to the simple polygon, the visibility polygon of $q$ can be queried as the simple polygon is updated with vertex insertions and deletions. 
We also devised dynamic algorithms to query for the weak visibility polygon of a line segment $s$ when $s$ is located in the simple polygon. 
The query time complexity of the proposed query algorithms', and the update time complexity of the visibility polygon maintenance algorithms are output-sensitive.
To our knowledge, this is the first result to give the fully-dynamic algorithm for maintaining the visibility polygon of a fixed point located in the simple polygon.
In addition, we have devised an incremental algorithm to update the weak visibility polygon of a line segment $pq$ located interior to simple polygon, as vertices are added to that simple polygon.
Its time complexity to update the weak visibility polygon due to the insertion of a vertex $v$ to the simple polygon is expressed in terms of the sum of the number of updates required to the shortest path trees rooted at $p$ and $q$.
It would be interesting to explore devising a fully-dynamic algorithm for the weak visibility polygon maintenance whose update time complexity is expressed in terms of the number of combinatorial changes required to the weak visibility polygon being maintained.
Further, we see lots of scope for future work in devising dynamic algorithms in the context of visibility, art gallery, minimum link path, and geometric shortest path problems.

\ignore {
\vspace{-0.1in}

\subsection*{Acknowledgements}

The authors wish to acknowledge the anonymous reviewers for the valuable comments which helped in improving the quality of the paper. 
}

\bibliographystyle{plain}

\end{document}